\documentclass[11pt]{article}

\usepackage[margin = 1in]{geometry}
\usepackage{amsmath}
\usepackage{amssymb, amsthm}
\usepackage{algorithm}
\usepackage{algorithmic}
\usepackage{mathrsfs}
\usepackage{hyperref}
\usepackage{bm}
\usepackage{enumitem}
\usepackage[dvipsnames]{xcolor}

\definecolor{citecolor}{HTML}{00AA00}
\hypersetup{
	unicode = true,
	colorlinks = true,
	citecolor = citecolor,
	filecolor = blue,
	linkcolor = blue,
	urlcolor = blue,
	pdfstartview = {FitH},
}

\theoremstyle{plain}
\newtheorem{theorem}{Theorem}[section]
\newtheorem{lemma}[theorem]{Lemma}

\newtheorem{proposition}[theorem]{Proposition}

\theoremstyle{definition}

\newtheorem*{remark}{Remark}

\newtheorem*{example*}{Example}

\algsetup{linenodelimiter=.}

\newcommand{\ang}[1]{\{#1\}}

\newcommand{\tildO}{O\tilde{~}}

\newcommand{\Romnum}[1]{\uppercase\expandafter{\romannumeral #1}}

\DeclareMathOperator{\End}{{\rm End}} 
\DeclareMathOperator{\gal}{{\rm Gal}} 

\def\Q{\ensuremath{\mathbb{Q}}}

\def\Z{\ensuremath{\mathbb{Z}}}
\def\F{\ensuremath{\mathbb{F}}}
\def\K{\ensuremath{\mathbb{K}}}
\def\P{\ensuremath{\mathbb{P}}}
\def\MM{\ensuremath{\mathsf{M}}}

\newcommand{\D}{\Delta}
\newcommand{\ph}{(\phi/f)}

\title{Drinfeld Modules with Complex Multiplication, Hasse Invariants and Factoring Polynomials 
over Finite Fields}

\author{
	Javad Doliskani\thanks{Institute for Quantum Computing, University of Waterloo
	(\tt{javad.doliskani@uwaterloo.ca}).}
	\and
	Anand Kumar Narayanan\thanks{Laboratoire d'Informatique de Paris 6, Pierre et Marie Curie 
	University (\tt{Anand.Narayanan@lip6.fr})}
	\and
	\'Eric Schost\thanks{Cheriton School of Computer Science, University of Waterloo 
	(\tt{eschost@uwaterloo.ca})}
}

\begin{document}
\maketitle

\begin{abstract}
	We present a novel randomized algorithm to factor polynomials over a
	finite field $\F_q$ of odd characteristic using rank $2$ Drinfeld
	modules with complex multiplication. The main idea is to compute a
	lift of the Hasse invariant (modulo the polynomial $f \in \F_q[x]$
	to be factored) with respect to a random Drinfeld module $\phi$ with
	complex multiplication. Factors of $f$ supported on prime ideals
	with supersingular reduction at $\phi$ have vanishing Hasse invariant
	and can be separated from the rest.
	Incorporating a Drinfeld module analogue of Deligne's congruence, we
	devise an algorithm to compute the Hasse invariant lift, which
	turns out to be the crux of our algorithm. The resulting expected
	runtime of $n^{3/2+\varepsilon} (\log q)^{1+o(1)}+n^{1+\varepsilon} (\log
	q)^{2+o(1)}$ to factor polynomials of degree $n$ over $\F_q$ matches
	the fastest previously known algorithm, the Kedlaya-Umans
	implementation of the Kaltofen-Shoup algorithm.
\end{abstract}



\section{Introduction}

Drinfeld modules of rank two are often presented as an analogue over
function fields such as $\F_q(x)$ (for a prime power $q$) of elliptic
curves over $\Q$; following Drinfeld's original terminology, we will
often call them {\em elliptic modules} in this paper. In very concrete
terms, a Drinfeld module over $\F_q(x)$ is simply a ring homomorphism
$\phi$ (together with some mild assumptions) from $\F_q[x]$ to the
ring of skew polynomials $\F_q(x)\ang{\tau}$, where $\tau$ satisfies
the commutation relation $\tau u = u^q \tau$ for $u$ in $\F_q(x)$. The
rank of a Drinfeld module is the degree of $\phi_x$, where as
customary, we write $\phi_a:=\phi(a)$ for $a$ in $\F_q[x]$.

In this definition, one may replace $\F_q(x)$ by any other other field
$L$ equipped with a homomorphism $\F_q[x]\to L$, and in particular by
a finite field of the form $L=\F_q[x]/f$; one may then define the {\rm
  reduction} of a Drinfeld module over $\F_q(x)$ modulo an irreducible
$f \in \F_q[x]$. Then, there exist striking similarities between the
theory of elliptic curves over $\Q$ and their reductions modulo
primes, and that of elliptic modules over $\F_q(x)$ and their
reduction modulo irreducible polynomials $f$. For instance, such
notions as endomorphism ring, complex multiplication, Hasse
invariants, supersingularity, or the characteristic polynomial of the
Frobenius, \dots~can be defined in both contexts, and share many
properties. On the other hand, while the literature on algorithmic
aspects of elliptic curves is extremely rich, this is not the case for
Drinfeld modules; only recently have they been considered under the
algorithmic viewpoint (for instance, it is known that they are not
suitable for usual forms of public key cryptography~\cite{Scanlon01}).

In this article, we give an algorithm for the computation of the Hasse
invariant of elliptic modules over finite fields, and show how efficient
algorithms for this particular problem (and a natural generalization
thereof) can be used to factor polynomials over finite fields.

To wit, recall that the Hasse invariant $h_E$ of an elliptic curve $E:
y^2=x^3+Ax+B$ over a finite field $\F_p$, $p>2$, can be defined as the
coefficient of degree $p-1$ in $(x^3+Ax+B)^{(p-1)/2}$ (other
definitions set it to be $1$ if this coefficient is nonzero, $0$
otherwise). The definition of $h_E$ makes it possible to compute it
using a number of operations softly linear in $p$, but one can do
better: it is possible to compute $h_E$ without computing all previous
coefficients, using the fact that the coefficients of 
$(x^3+Ax+B)^{(p-1)/2}$ satisfy a linear recurrence with polynomial
coefficients, and applying techniques for such recurrences due to
Strassen~\cite{Strassen76} and Chudnovsky and
Chudnovsky~\cite{ChCh88} (see also~\cite{BoGaSc07}).

In the case of elliptic modules, we will consider $\phi: \F_q[x] \to
\F_q(x)\ang{\tau}$, such that $\phi(x)=x+ g_\phi \tau+\Delta_\phi
\tau^2$, with $g_\phi$ and $\Delta_\phi$ in $\F_q[x]$. If $f$ is
irreducible of degree $k$ and does not divide $\Delta_\phi$, we say
that $\phi$ has {\em good reduction} at $f$. Then, the {\em Hasse
  invariant} $h_{\phi,f}$ is defined as the coefficient of $\tau^k$ in
$\phi_f\bmod f =\sum_{i=0}^{2k} h_i \tau^i$, with $h_i$ in $\F_q[x]/f$
for all $i$ (all coefficients of index less than $k$ vanish modulo
$f$).  For an arbitrary squarefree $f$, using a recurrence due to
Gekeler~\cite{gek} (which is somewhat similar to the one used for
elliptic curves), it is possible to define a related quantity which we will write
$\bar h_{\phi,f}$, in a way that ensures arithmetic properties needed
for the factorization algorithm below (see Eq.~\eqref{eqdef:hbar} for the 
definition).

As in the case of elliptic curves, one can deduce a straightforward
algorithm from the definition of $\bar h_{\phi,f}$; our first
contribution is to show that a better algorithm exist. We will
consider two different complexity models in our runtime analysis: an
\textit{algebraic model} and a \textit{boolean model}. In the former
we count the number of operations $+, \times, \div$ in the field
$\F_q$, while in the latter we count the number of bit operations,
over a standard RAM (the main reason behind this dichotomy is that
some operation at the core of our algorithms, namely {\em modular
  composition}, admits faster algorithms in the boolean model; we
discuss this further in the next section).  We denote by $\MM$ a
function such that polynomials of degree $n$ over any ring can be
multiplied in $\MM(n)$ base ring operations, and such that the
superlinearity conditions of~\cite[Chapter~8]{vzGG} are satisfied; we
can take $\MM(n) \in O(n\log n \log\log n)$~\cite{CaKa91}. We let
$\omega$ be a feasible exponent for square matrix multiplication; we
can always take $\omega \le 2.38$ \cite{CoWi90,LeGall14}. Then, our
first main result is the following theorem.

\begin{theorem}\label{theo:1}
  Let $\phi: \F_q[x] \to \F_q(x)\ang{\tau}$ be an elliptic
  module, such that $\phi(x)=x+ g_\phi \tau+\Delta_\phi \tau^2$, with
  $g_\phi$ and $\Delta_\phi$ in $\F_q[x]$.
  Given a squarefree polynomial $f$ of degree $n$, as well as $g_\phi
  \bmod f$ and $\Delta_\phi \bmod f$, one can compute $\bar h_{\phi,f}$ using
  \[O(n^{(1 - \beta)(\omega - 1) / 2 + (\omega + 1) / 2} + \MM(n^{1+\beta})\log qn)\] 
  operations in $\F_q$, for any $\beta \in (0,1)$, or using
  \[O( n^{3/2+\varepsilon} (\log q)^{1+o(1)} + n^{1+\varepsilon} (\log q)^{2+o(1)})\]
  bit operations, for any $\varepsilon > 0$.
\end{theorem}
The algorithm is inspired by the baby steps / giant steps algorithms
for recurrences with polynomial coefficients
of~\cite{Strassen76,ChCh88,BoGaSc07}, and also borrows heavily from
Kaltofen and Shoup's baby steps / giant steps distinct degree
factorization algorithm~\cite{ks}; indeed, the structures of the
algorithms are similar, and the runtime reported here is the same as that
in that reference. As in~\cite{ks}, for our first claim, taking
$\omega \approx 2.375$ and $\beta \approx 0.815$, the complexity is
$O(n^{1.815}\log q)$ operations in $\F_q$.

We use these results to design a polynomial factorization
algorithm. The resulting runtime is not better than that in~\cite{ks}
(in the algebraic model) or~\cite{ku} (in the boolean model); however,
we believe it is worth stating such results, since they bring a new
perspective to polynomial factorization questions.

The use of Drinfeld modules for polynomial factorization actually
goes back to work of Panchishkin and Potemine \cite{pp}, whose algorithm was
rediscovered by van der Heiden \cite{vdH}. These algorithms, along
with the second author's Drinfeld module black box Berlekamp
algorithm~\cite{nar} are in spirit Drinfeld module analogues of
Lenstra's elliptic curve method to factor integers \cite{len}. The
Drinfeld module degree estimation algorithm of \cite{nar} uses
Euler-Poincar\'e characteristics of Drinfeld modules to estimate the
factor degrees in distinct degree factorization. A feature common to
the aforementioned algorithms is their use of random Drinfeld modules,
which typically don't have complex multiplication.

We take a different approach. To factor a squarefree polynomial $f$,
we construct a random elliptic module $\phi$ with complex
multiplication by an imaginary quadratic extension of the rational
function field $\F_q(x)$ with class number $1$. At roughly half of the
prime ideals $\langle g \rangle$ in $\F_q[x]$, $\phi$ has so-called
supersingular reduction. Concretely, the properties of the quantity
$\bar h_{\phi,f}$ computed modulo $f$ imply that for any prime factor
$g$ of $f$, $g$ is a prime with supersingular reduction for $\phi$ if
and only if $\bar h_{\phi,f} = 0 \bmod g$. The construction of $\phi$
ensures that this happens with non-vanishing probability. Altogether,
we obtain the following result.

\begin{theorem}
  \label{theo:main-factor}
  Suppose that given any squarefree polynomial $f$ of degree $n$ over
  $\F_q$, and any Drinfeld module $\phi$ as in Theorem~\ref{theo:1},
  one can compute $\bar h_{\phi,f}$ in $H(n,q)$ operations in $\F_q$,
  resp.\ $H^*(n,q)$ bit operations. Suppose also that
  $H(n_1,q)+H(n_2,q) \le H(n_1+n_2,q)$, resp.\ $H^*(n_1,q)+H^*(n_2,q)
  \le H^*(n_1+n_2,q)$, holds for all $n_1,n_2 \ge 0$.  Then there is a
  randomized algorithm that can factor degree $n$ polynomials over
  $\F_q$ in an expected
  \begin{itemize}
	\item $\tildO(H(n,q)  + n^{(\omega+1)/2} + \MM(n)\log q )$ operations in $\F_q$, or
	\item $\tildO(H^*(n,q)) + n^{1+\varepsilon}(\log q)^{2+o(1)}$ bit operations.
  \end{itemize}
\end{theorem}
The runtime reported in the first item in Theorem~\ref{theo:1} imply
that we can take $H(n,q) = \kappa n^{1.815}\log q$, for a suitable
constant $\kappa$. Since such a function satisfies the superlinearity
assumption of Theorem~\ref{theo:main-factor}, the resulting runtime for the factoring
algorithm is is $O(n^{1.815}\log q)$ operations in $\F_q$.  In a
boolean model, the second item in Theorem~\ref{theo:1} shows that for
any $\varepsilon > 0$, we can take $H^*(n,q) = \kappa^*_\varepsilon (
n^{3/2+\varepsilon} (\log q)^{1 +o(1)} + n^{1+\varepsilon}(\log
q)^{2+o(1)})$ bit operations, for some constant $\kappa^*_\varepsilon$; superlinearity still holds,
which then implies a similar runtime for factorization. 

As mentioned above, these results do not improve on the
state-of-the-art for polynomial factorization.
Since Berlekamp's randomized polynomial time algorithm~\cite{ber},
polynomial factorization over finite fields has been the subject of a
vast body of work. Among important milestones, we
mention~\cite{cz,gs}; the first subquadratic algorithms are due to
Kaltofen and Shoup~\cite{ks}, with runtimes $O(n^{1.815}\log q)$
operations in $\F_q$  (two algorithms are in that reference:
a fast distinct degree factorization algorithm, and an algorithm
derived from the Berlekamp algorithm).

Kaltofen and Shoup already pointed out that a quasi-linear time
algorithm for modular composition would yield an exponent $3/2$ for
polynomial factorization; Kedlaya and Umans~\cite{ku} showed that
modular composition can be done in time $n^{1+\varepsilon} (\log
q)^{1+o(1)}$ in a boolean model, and exhibited the resulting
algorithm for factoring polynomials, with runtime $n^{3/2+\varepsilon}
(\log q)^{1 +o(1)} + n^{1+\varepsilon}(\log q)^{2+o(1)}$ bit
operations. 

We remark that our algorithm has the property of not requiring
separate distinct degree and equal degree factorization phases. In the
algebraic model, Kaltofen and Shoup's black box Berlekamp algorithm
has a similar runtime and also bypasses the distinct degree / equal
degree factorization stages; however, we are not aware of a version of
that algorithm that would feature a runtime of $n^{3/2+\varepsilon} (\log q)^{1 +o(1)} +
n^{1+\varepsilon}(\log q)^{2+o(1)}$ in the boolean model.
As to if the exponent $3/2$ in $n$ can be lowered, this remains an
outstanding open question.

The paper is organized as follows. In \S\ref{drinfeld_section},
elliptic modules are introduced; we define Hasse invariants and give
our algorithm for their computation, proving Theorem~\ref{theo:1}.  In
\S\ref{randomized_section}, we describe our factoring algorithm and
prove Theorem~\ref{theo:main-factor}, using a few notions from
function field arithmetic.


\section{Computing Hasse invariant lifts of elliptic modules}\label{drinfeld_section}


\subsection{Basic definitions}

Let $\F_q[x]$ denote the polynomial ring in the indeterminate $x$, for
some odd prime power $q$, and let $L$ be a field equipped with a
homomorphism $\gamma:\F_q[x] \to L$; typical examples for us will be
$L=\F_q(x)$ and $L=\F_q[x]/f$, for some irreducible polynomial $f$ in
$\F_q[x]$.  Let us further consider the skew polynomial ring
$L\{\tau\}$, where $\tau$ satisfies the commutation rule $\forall u
\in L, \tau u = u^q \tau$. Given an integer $r > 0$, a Drinfeld module of rank $r$
over $L$ is a ring morphism
\begin{equation}
\label{equ:Drinfeld}
	\begin{array}{rrll}
		\phi : & \F_q[x] & \longrightarrow & L\{\tau\} \\
		& x & \longmapsto & a_0 + a_1\tau + \cdots + a_r\tau^r	
	\end{array}
\end{equation}
with $a_0 = \gamma(x)$ and $a_r \ne 0$.

An \textit{elliptic module} is a rank-2 Drinfeld module, obtained by
setting $r = 2$ in \eqref{equ:Drinfeld}. Consider such an elliptic module
over $L=\F_q(x)$, and let $\gamma$ be the inclusion $\F_q[x] \to \F_q(x)$.
In such a case, $\phi$ will be written as
\[
\begin{array}{rrll}
	\phi : & \F_q[x] & \longrightarrow & \F_q(x)\ang{\tau} \\
	& x & \longmapsto & x + g_\phi \tau + \Delta_\phi \tau^2	
\end{array}
\]
for some $g_\phi \in \F_q[x]$ and nonzero $\Delta_\phi \in \F_q[x]$. 
For an irreducible polynomial $f \in \F_q[x]$, if $\Delta_\phi$ is
nonzero modulo $f$, then the reduction $\phi/f$ of $\phi$ at $f$ is defined as the elliptic module
\[
\begin{array}{rcll}
	\phi/f : & \F_q[x] & \longrightarrow & (\F_q[x]/f) \ang{\tau} \\
	& x & \longmapsto & x + (g_\phi \bmod f) \tau + (\Delta_\phi\bmod f) \tau^2;
\end{array}
\]
we say that $\phi$ has {\em good reduction} at $f$ in this case.
Then,
 the image of $a \in \F_q[x]$ under $\phi/f$ is denoted by
$(\phi/f)_a$. Even if $\Delta_{\phi}$ is zero modulo $f$, one could
still obtain the reduction $(\phi/f)$ of $\phi$ at $f$ through
minimal models of $\phi$~\cite{gek1}; we refrain from
addressing this case since our algorithms do not require it.

Let $\phi$ be as above and let $f \in \F_q[x]$ be an irreducible 
polynomial not dividing $\Delta_\phi$. The {\em Hasse invariant} $h_{\phi,f}
\in \F_q[x]/f$ of $\phi$ at $f$ is the coefficient of $\tau^{\deg(f)}$
in the expansion $$\ph_f = \sum_{i=0}^{2\deg(f)} h_i \tau^i \in
(\F_{q}[x]/f)\ang{\tau}.$$
The elliptic module $\phi$ has {\em supersingular reduction} at $f$ if
$h_{\phi,f}$ vanishes~\cite{gos}; otherwise, we say that it has {\em
  ordinary reduction}. If the choice of $\phi$ is clear from context,
we will simply call $f$ supersingular.

Recursively define a sequence $(r_{\phi,k})_{k \in \mathbb{N}}$ in
$\F_q[x]^\mathbb{N}$ as $r_{\phi,0}:=1$, $r_{\phi,1}:=g_\phi$ and for
$k>1$,
\begin{equation}\label{eisenstein_recurrence}
r_{\phi,k} := g_\phi^{q^{k-1}}r_{\phi,k-1} - (x^{q^{k-1}}-x)\D_\phi^{q^{k-2}} r_{\phi,k-2} \in \F_q[x].
\end{equation}
Gekeler~\cite[Eq 3.6, Prop 3.7]{gek} showed that $r_{\phi,k}$ is
the value of the normalized Eisenstein series of weight $q^{k}-1$ on
$\phi$ and established Deligne's congruence for Drinfeld modules,
which ascertains that for any irreducible $f$ of degree $k \geq 1$ with
$\Delta_\phi \neq 0 \mod f$, we have
\begin{equation}\label{deligne_congruence}
 h_{\phi,f} = r_{\phi,k} \mod f.
\end{equation}
Hence $r_{\phi,k}$ is in a sense a lift to $\F_q[x]$ of all the Hasse
invariants of $\phi$ at primes of degree $k$. Using the sequence
$r_{\phi,k}$ allows us to define a polynomial $\bar h_{\phi,f}$ 
for arbitrary non-irreducible polynomials $f$, by setting 
\begin{equation}\label{eqdef:hbar}
\bar h_{\phi,f} :=
\gcd\left(r_{\phi,\deg(f)}\bmod f,r_{\phi,\deg(f)+1}\bmod f\right).
\end{equation}
Our definition will be justified in light of Lemma~\ref{lemma:split};
note that unlike in the irreducible case, we define $\bar h_{\phi,f}$ only up to
a non zero $\F_q$ multiple, which will suffice for our purpose. The
following subsections give our algorithm to compute $\bar h_{\phi, f}$ for
an arbitrary squarefree $f$, thereby proving Theorem~\ref{theo:1}.


\subsection{Some key subroutines}

 We summarize here the main results we will need;
most of them are well-known, and originate from work of von zur
Gathen-Shoup~\cite{gs} or Kaltofen-Shoup~\cite{ks}. A key ingredient
is the use of {\em modular composition}, that is, the operation
$(f,g,h)\mapsto f(g) \bmod h$, since several operations related to the
Frobenius map can be computed efficiently using this as a subroutine.

For $f,g,h$ of degree $n$, the best known algorithm for modular
composition was for long Brent and Kung's result~\cite{BrKu78}, with a
cost of $O(n^{(\omega+1)/2})$ base field operations; improvements
using fast rectangular matrix multiplication followed in the work of
Huang and Pan~\cite{HuPa98}. More recently, Kedlaya and Umans showed
that in a {\em boolean} model, there exist algorithms of cost close to
{\em linear}: for any $\varepsilon > 0$, there is an algorithm for
modular composition of degree $n$ polynomials over $\F_q$ that takes
$n^{1+\varepsilon} (\log q)^{1+o(1)}$ bit operations.

As of now, algorithms based on Brent and Kung's result have been
implemented on a variety of platforms~\cite{shoup2001ntl,BoCaPl97,Hart10},
and still outperform implementations of the Kedlaya-Umans algorithm. As
a result, we decided to give two variants of our main algorithm, using
these two possible key subroutines.

For the rest of this section, consider a squarefree polynomial $f$ in
$\F_q[x]$, define $\K=\F_q[x]/f$ and let $\xi$ be the image of $x$ in
$\K$.  Let $\tau: \K \to \K$ be the $\F_q$-linear $q$th-power
Frobenius map $a \mapsto a^q$; since $f$ is squarefree, $\tau$ is
invertible. It will then be convenient to define $\xi_i:=\tau^i(\xi)$,
for $i$ in $\Z$; then, for any $a$ in $\K$ and $i$ in $\Z$, $a(\xi_i)$
is well-defined (as $A(\xi_i) \in \K$, for an arbitrary lift $A$ of $a$ to
$\F_q[x]$) and satisfies $a(\xi_i)=\tau^i(a)$.  In particular, for any
$i,j$ in $\Z$, the relation $\xi_i(\xi_j)=\xi_{i+j}$ holds.

The following items describe subroutines needed in our main
algorithm. With the partial exception of the last one (see the remark below), all of them are known results.
\begin{enumerate}
\item[{\bf 1.}] Computing $\xi_1=\xi^q$ is done by repeated squaring, using
  $O(\MM(n)\log q)$ operations in $\F_q$.
\item[{\bf 2.}] Once $\xi_1=\xi^q$ is known, any $\xi_r$, $r \ge 0$, can be
  computed at the cost of $O(\log r)$ modular compositions, using the
  relation $\xi_i(\xi_j) = \xi_{i+j}$ stated above.
\item[{\bf 3.}] Given $\xi_i$ ($i \ge 0$), we can compute $\xi_{-i}$ by solving
  the linear equation $\xi_{-i}(\xi_i) =\xi$. This can be done using
  {\em transposed} modular composition, with a cost identical (up to a
  constant factor) to that of modular composition 
  itself~\cite{Shoup94,DeDoSc2014}.
\item[{\bf 4.}] Given $(a_0,\dots,a_{k-1})$ in $\K$ and $\xi_i$, for
  some $i \in \Z$, consider the question of computing
  $(\tau^i(a_0),\dots,\tau^i(a_{k-1}))=(a_0(\xi_i),\dots,a_{k-1}(\xi_i))$. In
  the boolean model, the cost is $k n^{1+\varepsilon} (\log
  q)^{1+o(1)}$ bit operations, which is essentially optimal.
  In the algebraic model, for $k=O(n)$, Lemma~3 in~\cite{ks} shows how
  to do this in $O(k^{(\omega-1)/2}n^{(\omega+1)/2})$ operations in
  $\F_q$ instead of $O(kn^{(\omega+1)/2})$, by exploiting the fact that the second argument is the same
  for all instances.

\item[{\bf 5.}] Given integers $u,\ell,m \in O(n)$, and $\xi_1$, one
  can compute $(\xi_{-u},\xi_{-(u+\ell)},\dots,\xi_{-(u+(m-1)\ell)})$
  using $O((m^{(\omega-1)/2}+\log n) n^{(\omega+1)/2})$ operations in $\F_q$,
  resp.\, $m n^{1+\varepsilon} (\log q)^{1+o(1)}$ bit operations. The
  latter estimate is straightforward; for the former, we
  follow~\cite[Lemma~4]{ks}, which shows (in our notation) how to
  compute $\xi_\ell,\xi_{2\ell},\dots,\xi_{(m-1)\ell}$. We work here
  with negative indices, but this hardly changes the procedure.

  We first compute $\xi_{-\ell}$, at the cost of $O(\log \ell)$
  modular compositions. Then, for $k \ge 1$, assuming that we know
  $\xi_{-\ell},\xi_{-2\ell},\dots,\xi_{-k \ell}$, we deduce
  $\xi_{-(k+1)\ell},\xi_{-(k+2)\ell},\dots,\xi_{-2k \ell}$ from the
  relation $\xi_{-(k+j)\ell}=\xi_{-j\ell}(\xi_{-k\ell})$, by means of
  item ${\bf 4}$. We repeat this process for $k=1,2,4,\dots$, stopping
  at the first power of two greater than or equal to $m$.  At this
  stage, we know $\xi_{-\ell},\dots,\xi_{-(m-1)\ell}$. To conclude, we
  compute $\xi_{-u}$ at the cost of $O(\log u)$ modular compositions, and finally
  $\xi_{-\ell}(\xi_{-u}),\dots,\xi_{-(m-1)\ell}(\xi_{-u})$. The total
  cost is $O(m^{(\omega-1)/2}n^{(\omega+1)/2})$ operations in $\F_q$.

\item [{\bf 6.}] For $u,\ell,m$ as above, given $(a_0,\dots,a_{m-1})$
  and $\xi_1$, we finally show how to compute the product
  $\tau^{u+(m-1)\ell}(a_{m-1}) \cdots \tau^{u}(a_0).$ Here, we
  actually take all $a_i$ as $2\times 2$ matrices over $\K$, since
  this is what we will need below (this has no impact on the algorithm
  description, except that we must account for the non-commutativity
  of the product). As above, in the boolean model, the cost is easily
  seen to be $m n^{1+\varepsilon} (\log q)^{1+o(1)}$ bit operations.
  
  To discuss the algorithm in the algebraic model, without loss of
  generality, we assume that $m$ is a power of $2$, say
  $m=2^t$. First, we replace our input by
  $(a_0(\xi_u),\dots,a_{m-1}(\xi_u))$. Let us then set $\mu=m$ and
  $\zeta=\xi_\ell$.  For $k=0,\dots,t-1$, we do the following: for
  $i=0,\dots,\mu/2-1$, replace $a_i$ by $a_{2i+1}(\zeta) a_{2i}$; then
  let $\zeta=\zeta(\zeta)$ and $\mu=\mu/2$.  At the end of the loop,
  the first entry in the sequence is the requested output.

  Computing $\xi_\ell$ and $\xi_u$ takes $O(\log n)$ modular
  compositions.  The initial composition by $\xi_u$ takes time
  $O(m^{(\omega-1)/2}n^{(\omega+1)/2})$, by item ${\bf 4}$; for an
  index $k$ in the main loop, the cost is similarly
  $O((m/2^k)^{(\omega-1)/2}n^{(\omega+1)/2})$; The overall runtime is
  thus $O((m^{(\omega-1)/2}+\log n)n^{(\omega+1)/2})$
  operations in $\F_q$.
\end{enumerate}

\begin{remark}
  The description we give for item ${\bf 6}$ answers a question
  in~\cite[Section~3.2]{ks}: the authors proved the existence of an
  algorithm to compute (in our notation) an expression of the form
  $a_0 + \tau^\ell(a_1) + \cdots + \tau^{(m-1)\ell}(a_{m-1})$
in
  $O(m^{(\omega-1)/2}n^{(\omega+1)/2})$ base field operations, but left
  the actual description of such an algorithm as an open question
  (in that context, $\log n$ is negligible in front of $m^{(\omega-1)/2}$,
  and the $a_i$'s are actually in $\K$).
  The procedure we gave above can be adapted to the computation
  of such a sum, simply by replacing all products by additions 
  in the main loop.
\end{remark}


\subsection{Efficient computation of $\bar h_{\phi,f}$}

Given a squarefree $f$ in $\F_q[x]$ as above, with $\deg(f)=n$, and an
elliptic module $\phi$ over $\F_q(x)$, we present an efficient
algorithm to compute $r_{\phi,n} \bmod f$ and $r_{\phi,n+1} \bmod f$,
thereby yielding $\bar h_{\phi,f} = \gcd(r_{\phi,n}\bmod
f,r_{\phi,n+1}\bmod f)$. The latter gcd can be computed in quasi-linear 
time, namely $O(\MM(n)\log(n))$ operations in $\F_q$, so we will not 
discuss it further.

Our strategy to compute $r_{\phi,n} \bmod f$ and $r_{\phi,n+1} \bmod
f$ is to first phrase the recurrence in matrix form with entries being
polynomials modulo $f$. We then observe that solving the recurrence
amounts to computing the product of a carefully constructed sequence
of matrices twisted by the Frobenius action. The final step is to
construct a polynomial with matrix coefficients, whose evaluations
allow us to rapidly compute the aforementioned product; the
evaluations are then computed using a fast multipoint evaluation
algorithm.

The overall structure of the algorithm is similar to the distinct
degree factorization algorithm of~\cite{ks}, with the polynomial
matrix described above being akin to the ``interval polynomial'' used
in that algorithm.  Another inspiration for our algorithm is
computation with simpler recurrences, that go back to Strassen's
deterministic integer factorization algorithm~\cite{Strassen76} and
Chudnovsky and Chudnovsky's algorithm for linear recurrences with
polynomial coefficients~\cite{ChCh88}. These baby-steps / giant steps
techniques allow one to compute the $n$th term in a sequence defined
by such a recurrence, for instance $u_0=1, u_{n+1}=(n+1) u_n$, in
$\tildO(\sqrt{n})$ base field operations. These ideas were
subsequently refined and applied to the computation of the Hasse-Witt
matrix of hyperelliptic curves in~\cite{BoGaSc07}.

With $\K = \F_q[x]/f$ and $\xi$ as before, we have to compute $\rho_n:=
r_{\phi,n}(\xi)$
and $\rho_{n+1}:=
r_{\phi,n+1}(\xi)$, where the sequence $r_{\phi,k}$ is
from~\eqref{eisenstein_recurrence}; equivalently,
\begin{equation}\label{eq:rho}
  \rho_0=1,\quad \rho_1=\gamma,\quad
\rho_{k} = \gamma^{q^{k-1}}\rho_{k-1} -
(\xi^{q^{k-1}}-\xi)\delta^{q^{k-2}} \rho_{k-2},
\end{equation}
with $\gamma = g_\phi(\xi)\in \K$ and $\delta=\Delta_\phi(\xi) \in
\K$.  Computing all terms $\rho_0,\rho_1,\dots,\rho_{n+1}$ takes
$\Omega(n^2)$ operations in $\F_q$, since merely writing down each
$\rho_i$ involves $\Theta(n)$ operations. The algorithm in this
section takes subquadratic time.

In the case of elliptic curves, the Hasse invariant is obtained as the
element of index $(p-1)/2$ in an order-2 recurrence with polynomial
coefficients. Seeing the similarity with Eq.~\eqref{eq:rho}, it is
natural to adapt these ideas to our context. This is however not
entirely straightforward.  Indeed, given $\rho_{k-1}$ and
$\rho_{k-2}$, computing $\rho_k$ now boils down to applying powers of the
Frobenius endomorphism, together with a few polynomial multiplications
/ additions.

The recurrence~\eqref{eq:rho} can be written as
\[
\begin{bmatrix}
\rho_{k - 1} \\
\rho_{k} 
\end{bmatrix} = 
\begin{bmatrix}
0 & 1 \\
(\xi-\xi^{q^{k - 1}})\delta^{q^{k - 2}} & \gamma^{q^{k - 1}}
\end{bmatrix}
\begin{bmatrix}
\rho_{k - 2} \\
\rho_{k - 1} 
\end{bmatrix}.
\]
Let as before $\tau: \K \to
\K$ be the $\F_q$-linear $q$th-power Frobenius map,
and define the following sequence of matrices in 
$\mathscr{M}_2(\K)$:
\[
A_k :=\begin{bmatrix}
0 & 1 \\
(\xi-\xi^{q^{k + 1}})\delta^{q^{k }} & \gamma^{q^{k + 1}}
\end{bmatrix} =
\begin{bmatrix}
0 & 1 \\
(\xi-\tau^{k+1}(\xi)) \tau^{k}(\delta) &  \tau^{k+1}(\gamma)
\end{bmatrix}
;
\]
then, we have
\[
\begin{bmatrix}
\rho_{k - 1} \\
\rho_{k} 
\end{bmatrix} = 
A_{k-2}A_{k - 3} \cdots A_0
\begin{bmatrix}
1 \\
\gamma
\end{bmatrix}.
\]
Given integers $m \le m'$, we show how to compute the product 
\[B_{m,m'} := A_{m' - 1} \cdots A_m ~ \in ~ \mathscr{M}_2(\K),\]
for then we can read off $\rho_{n}$ and $\rho_{n+1}$ from $B_{0,n}\left [\begin{smallmatrix}
1 \\
\gamma \\
\end{smallmatrix} \right ]$. We need the extra flexibility of 
starting the product at index $m$ in the algorithm; concretely, we will
rely on the relation $B_{m',m''}B_{m,m'}=B_{m,m''}$, for any
integers $m \le m' \le m''$. 

Extend the mapping $\tau$ to the (non-commutative) polynomial ring
$\mathscr{M}_2(\K)[Y]$ by leaving $Y$ fixed and acting on the
coefficient matrices entry-wise.  Let further
\[
\mathcal{A} := 
\begin{bmatrix}
0 & 1 \\
-\tau(\xi)\delta & \tau(\gamma)
\end{bmatrix}
+
\begin{bmatrix}
0 & 0 \\
\delta & 0
\end{bmatrix} Y ~ \in ~ \mathscr{M}_2(\K)[Y],
\]
and for $\mathcal{M} \in \mathscr{M}_2(\K)[Y]$ and $\zeta \in \K$, let $\mathcal{M}(\zeta)$ 
denote the image of $\mathcal{M}$ under the substitution 
\[Y \longmapsto 
D_\zeta=\begin{bmatrix}
	\zeta & 0 \\
	0 & \zeta
\end{bmatrix};
\]
in particular, $\mathcal{M} \mapsto \mathcal{M}(\zeta)$ is a ring
homomorphism, since $D_\zeta$ is in the center of $\mathscr{M}_2(\K)$.
Then, for any $k \ge 0$, we have $$A_k = \tau^{k}(\mathcal{A})(\xi).$$
Choose an integer $\ell \le n+1$ and define
\[\mathcal{B} := \tau^{\ell-1}(\mathcal{A}) \cdots \tau(\mathcal{A}) \mathcal{A}  ~ \in ~ \mathscr{M}_2(\K)[Y],\]
so that in particular
$\mathcal{B}(\xi) = A_{\ell-1} \cdots A_1 A_0.$ More generally, 
we use the fact that for all integers $i, j$, we have
\[A_{i + j} = \tau^{i + j}(\mathcal{A})(\xi) = \tau^i\big(\tau^j(\mathcal{A})( 
\xi_{-i})\big)\]
to deduce that  for all $i \ge 1$,
\[\tau^{i}(\mathcal{B}( \xi_{-i})) = A_{i + \ell-1} \cdots A_{i + 1}A_{i} = B_{i,i+\ell}.\]
In particular, for any integer $m$, $B_{0,m\ell}$ can be computed as the product of the  matrices
\[
\tau^{(m-1) \ell}\big(\mathcal{B} ( \xi_{-(m-1)\ell} ) \big), \dots, 
\tau^{\ell}\big(\mathcal{B} ( \xi_{-\ell}) \big), \mathcal{B} \big(\xi\big);
\]
more generally, for any integers $m,u$, we have 
\[
B_{u,u+m\ell}=
\tau^{u+(m-1) \ell}\big(\mathcal{B} ( \xi_{-(u+(m-1) \ell)}) \big) \cdots
\tau^{u+\ell}\big(\mathcal{B} ( \xi_{-(u+\ell)}) \big)
\tau^{u}\big(\mathcal{B} ( \xi_{-u}) \big).
\]
Suppose that we know $B_{0,u}$, for some $u < n$. As in~\cite{ks},
let $\beta$ be an arbitrary constant in $(0,1)$, and define $\ell =
\lceil (n-u)^\beta \rceil$, $m = \lfloor (n-u) / \ell
\rfloor$. Then, $v=u+m\ell$ satisfies $v \le n$ and $n-v \le
(n-u)^\beta$.  The discussion above suggests Algorithm
\ref{alg:hasse-inv} for computing $B_{u,v}$, from which we can
deduce $B_{0,v} = B_{u,v}B_{0,u}$.

\begin{algorithm}[H]
  \caption{Main subroutine for the Hasse invariant}
  \label{alg:hasse-inv}
  \begin{algorithmic}[1]
    \REQUIRE $f$ squarefree in $\F_q[x]$,  $\delta, \gamma \in \K=\F_q[x]/f$, integers $u,n$, with $u < n$
    \ENSURE $B_{u,v}$ as defined above
    \STATE Let $\ell := \lceil (n-u)^\beta \rceil$ and $m := \lfloor (n-u) / \ell \rfloor$
    \STATE\label{step:hasse-2}
    Compute $\mathcal{B} = \tau^{\ell-1}(\mathcal{A}) \cdots \tau(\mathcal{A}) \mathcal{A}$
    \STATE\label{step:hasse-3}
    Compute $\xi_{-(u+i\ell)} = \tau^{-(u+i\ell)}(\xi)$ for $0 \le i < m$
    \STATE\label{step:hasse-4}
    Compute $\beta_i = \mathcal{B}(\xi_{-(u+i\ell)})$ for $0 \le i< m$
    \RETURN\label{step:hasse-5}  the product $\tau^{u+(m-1)\ell}(\beta_{m-1})\cdots \tau^u(\beta_0)$
  \end{algorithmic}
\end{algorithm}

Correctness of the algorithm follows from the preceding remarks.  We
will give two different runtime analyses, in respectively the
algebraic and boolean model, where the main difference lies in the
cost of modular composition.

\begin{itemize}
\item In the whole algorithm, we compute once $\xi_1=\tau(\xi)$; the
  cost is $O(\MM(n)\log(q))$ operations in $\F_q$.

\item A first solution for Step \ref{step:hasse-2} is to compute $
   \tau(\mathcal{A}), \dots,\tau^{\ell-1}(\mathcal{A})$
  (using $\ell-1$ successive applications of the Frobenius) and
  multiply the results. The successive applications of $\tau$ are
  done using repeated squaring, and the subsequent multiplications
  using a $2 \times 2$ matrix version of the subproduct
  tree~\cite[Chapter~10]{vzGG}.  Multiplication of $2 \times 2$
  polynomial matrices in degree $k$ over $\K$ takes $O(\MM(k n))$
  operations in $\K$ (using Kronecker substitution), so the cost is
  $O(\ell\MM(n)\log q + \MM(\ell n)\log \ell)$ operations in $\F_q$.

  An alternative is to perform Step \ref{step:hasse-2}
  recursively: given $\tau^{i-1}(\mathcal{A})\cdots \mathcal{A}$, it
  takes one application of $\tau^{i}$ (resp.\ of $\tau$) and one
  matrix multiplication in $\mathscr{M}_2(\K)[Y]$ to compute
  $\tau^{2i-1}(\mathcal{A})\cdots \mathcal{A}$,
  resp.\ $\tau^{i}(\mathcal{A})\cdots \mathcal{A}$. 
  The power-of-$\tau$ map is computed using the iterated Frobenius
  algorithm of von zur Gathen and Shoup~\cite{gs}. In the algebraic
  model, the cost of modular composition makes this solution inferior
  to the one in the previous paragraph. In the boolean model, since $\xi^q$
  is known, the cost becomes $\ell n^{1+\varepsilon} (\log q)^{1+o(1)}$
  bit operations.
  
\item Items ${\bf 5}$ and ${\bf 6}$ in the previous subsection show
  that Step~\ref{step:hasse-3} and~\ref{step:hasse-5} can be done
  using $O(m^{(\omega-1)/2} n^{(\omega+1)/2})$ operations in $\F_q$,
  resp.\ $m n^{1+\varepsilon} (\log q)^{1+o(1)}$ bit operations.
  
\item Step \ref{step:hasse-4} can be done using multipoint evaluation
  \cite{vzGG}. We are evaluating a $2 \times 2$ polynomial matrix of
  degree at most $\ell$ at $m$ points.  We first build the subproduct
  tree at the given points, then reduce each entry of the matrix
  modulo the root polynomial of the tree, and apply the ``going down the tree''
  procedure of~\cite[Chapter~10]{vzGG}. Altogether, this takes
  $O(\MM(m n) \log m + \MM(\ell n))$ operations in $\F_q$.
\end{itemize}
Altogether, in the algebraic model, we obtain the following result.
\begin{proposition}
  \label{theo:hasse-inv}
  Algorithm \ref{alg:hasse-inv} can be implemented so as to run in
  \[O((n-u)^{(1 - \beta)(\omega - 1) / 2} n^{(\omega + 1) / 2} + \MM((n-u)^{\beta}n)\log qn)\] 
  operations in $\F_q$, where $\omega$ is the matrix multiplication exponent.
\end{proposition}
In the boolean model, we set $\beta = 1/2$. Then, each of the steps
\ref{step:hasse-2}, \ref{step:hasse-3}, \ref{step:hasse-4} and
\ref{step:hasse-5} can each be performed in $n^{3/2+\varepsilon} (\log
q)^{1+o(1)}$ bit operations, whereas the initial computation of
$\xi^q$ takes $ n^{1+\varepsilon} (\log q)^{2+o(1)}$ bit operations. To
summarize, we have the following result.
 
\begin{proposition}
  \label{theo:hasse-inv-bit}
  Algorithm \ref{alg:hasse-inv} can be implemented so as to run in
  \[O((n-u)^{1/2} n^{1+\varepsilon} (\log q)^{1+o(1)} + n^{1+\varepsilon} (\log q)^{2+o(1)})\]
  bit operations.
\end{proposition}

To conclude the proof of Theorem~\ref{theo:1}, we apply the previous
results with input $u=0$ and $n$, and obtain $B_{0,v}$, for $v$ as
defined previously. If $v=n$, we are done, otherwise we re-enter the
algorithm with input $v$ and $n$ (so that $n-v \le
n^\beta$), and so on. Starting from Proposition~\ref{theo:hasse-inv},
we see that the cost of the first call dominates the
overall runtime, so the first item in Theorem~\ref{theo:1} is proved.
Using Proposition~\ref{theo:hasse-inv-bit}, the situation is similar,
up to the total contribution of the second term in the sum, since a factor
$\log\log n$ appears; however, we can absorb it in the $1+\varepsilon$ 
exponent and the second claim in Theorem~\ref{theo:1} follows.


\section{A polynomial factoring algorithm}\label{randomized_section}

We can now give our elliptic module algorithm for polynomial
factorization. As input, we are given $f \in\F_q[x]$ of degree $n$;
without loss of generality, we may assume that $f$ is
squarefree~\cite{knu,yun}, that is, does not contain a square of an
irreducible polynomial as a factor.

Let $\phi: \F_q[x] \to \F_q(x)\ang{\tau}$ be an elliptic
module over $\F_q(x)$ and suppose that $\gcd(f,\Delta_\phi)=1$. 
 The following lemma is the key of our algorithm.
\begin{lemma}\label{lemma:split}
  Let $g$ be an irreducible factor of $f$. Then $\phi$ has supersingular
  reduction at $g$ if and only if $\bar h_{\phi,f} \bmod g=0$.
\end{lemma}
\begin{proof}
  First, note that by assumption, $\phi$ has good reduction at
  $g$. Let $k:=\deg(g)$ and $n:=\deg(f)$ and for $j \ge 0$ let
  $r_{\phi,j} \in \F_q[x]$ be the sequence as in
  recurrence~\eqref{eisenstein_recurrence}, defined with respect to
  $\phi$.
	
  Assume $\phi$ has supersingular reduction at $g$, that is,
  $h_{\phi,g}=0$ (recall that this is a polynomial defined modulo
  $g$). By Deligne's congruence, $r_{\phi,k} \bmod g = 0$.  Since $g$
  divides $x^{q^k}-x$, the recurrence
  relation~\eqref{eisenstein_recurrence} implies $r_{\phi,k+1} \bmod g
  = 0$. Since this recurrence has order~2, this further yields
  \begin{equation}\label{supersingular_zero}
    r_{\phi,j}  \bmod g= 0,\quad j \geq k.
  \end{equation}
  In particular, $r_{\phi,n} \bmod g= 0$ and $r_{\phi,n+1} \bmod g= 0$;
  this implies $\bar h_{\phi,f} \bmod g = 0$, since $g$ divides $f$.

  Conversely, assume $\bar h_{\phi,f}  \bmod g= 0$. To prove that
  $\phi$ has supersingular reduction at $g$, using Deligne's
  congruence, it suffices to show $r_{\phi,k} \bmod g=0$.  Since $\bar
  h_{\phi,f} \bmod g = 0$ and $g$ divides $f$, we have $r_{\phi,n}
  \bmod g = r_{\phi,n+1} \bmod g=0$. In particular, if $g=f$,
  we are done, since then $k=n$.

  Else, let $m$ be
  the smallest positive integer greater than $k$ such that $r_{\phi,m}
  \bmod g = 0$ and $r_{\phi,m+1} \bmod g =0$; such an $m$ exists, 
  since then $k < n$. From the
  recurrence \eqref{eisenstein_recurrence},
\[
r_{\phi,m+1}  \mod g = g_\phi^{q^m} r_{\phi,m} - \left(x^{q^m}-x\right) \Delta_{\phi}^{q^{m-1}} r_{\phi,m-1} \mod g;
\]
hence,
\[
\left(x^{q^m}-x\right) \Delta_{\phi}^{q^{m-1}} r_{\phi,m-1} \mod g = 0.
\]
Since $\phi$ has good reduction at $\phi$, we may conclude that either
$r_{\phi,m-1}\mod g=0$ or $k$ divides $m$. If $r_{\phi,m-1}=0 \bmod g$, then
by the minimality of $m$, we may conclude $m-1=k$, that is,
$r_{\phi,k}\mod g = 0$, proving our claim. If $k$ divides $m$, then
$m-k \geq k$ (since $m > k$). By~\cite[Lemma~2.3]{cor},
\[
r_{\phi,m}  \bmod g= r_{\phi,k} r_{\phi,m-k}^{q^k} \bmod g,\quad r_{\phi,m+1}  \bmod g = r_{\phi,k} r_{\phi,m+1-k}^{q^k} \bmod g,
\]
implying either $r_{\phi,k} \mod g = 0$ (proving our claim) or
$r_{\phi,m+1-k} \mod g =r_{\phi,m-k} \mod g =0$. Since $m-k
\geq k$, for the latter case to not contradict the minimality of $m$,
$m-k$ has to equal $k$, implying $r_{\phi,k} \mod g=0$, thereby proving
our claim.
\end{proof}

This suggests that we could use an elliptic module $\phi$ in a
polynomial factorization algorithm to separate supersingular primes
from those that are not. For most elliptic modules, the density of
supersingular primes is too small for this to work. However, for a
special class, elliptic modules with complex multiplication, the
density of supersingular primes is $1/2$.

Let $L=\F_q(x)(\sqrt{d})$ be a quadratic extension of $\F_q(x)$, for
some polynomial $d$ in $\F_q[x]$.  We say that $L$ is {\em imaginary}
if the prime $(1/x) \in \F_q(x)$ at infinity does not split in $L$;
this is in particular the case if $d$ is squarefree of odd degree (in which case the 
prime at infinity ramifies~\cite[Proposition~14.6]{Rosen02}).

An elliptic module $\phi$ over $\F_q(x)$ is said to have {\em complex
  multiplication} by an imaginary quadratic extension $L/\F_q(x)$ if
$\End_{\F_q(x)}(\phi)\otimes_{\F_q[x]} \F_q(x)$ is isomorphic to $L$,
where $\End_{\F_q(x)}(\phi)$ is the {\em ring of endomorphisms} of $\phi$,
that is, the ring of all elements of $\F_q(x)\{\tau\}$ that commute
with $\phi_x$. For $\phi$ with complex multiplication by $L/\F_q(x)$
and an irreducible polynomial $f \in \F_q[x]$ such that $\langle f
\rangle$ is unramified in $L/\F_q(x)$, $f$ is supersingular if and
only if it $\langle f \rangle$ is inert in $L/\F_q(x)$.

This suggests the following strategy to factor a monic squarefree
polynomial $f \in \F_q[x]$. Say $f$ factors into monic irreducibles as
$f = \prod_i f_i$. Pick an elliptic module $\phi$ with complex
multiplication by some imaginary quadratic extension $L/\F_q(x)$ and
compute $\bar h_{\phi,f}$. By Lemma~\ref{lemma:split}, we get that
\begin{equation}
\label{equ:fact-sep}
\gcd(\bar h_{\phi,f}, f) =
\prod_{\langle f_i \rangle \text{~inert in~} L/\F_q(x)}f_i
\end{equation}
is a factor of $f$. Since for every degree, roughly half the primes of
that degree are inert in $L/\F_q(x)$, the factorization thus obtained
is likely to be non trivial. Repeating the process for the resulting
factors leads to a complete factorization of $f$.\\

\noindent It remains to construct elliptic modules with complex
multiplication. Our strategy is to pick an $a \in \F_q$ at random and
construct an elliptic module $\phi$ with complex multiplication by the
imaginary quadratic extension $\F_q(x)(\sqrt{d})$ of discriminant
$d:=x-a$. From \cite{dor} (see also \cite[Theorem 6]{Schweizer97}), the
elliptic module $\phi^\prime$ with
\[ g_{\phi^\prime}(x):=\sqrt{d}+\sqrt{d}^q,\quad \D_{\phi^\prime} := 1 \]
has complex multiplication by $\F_q(x)(\sqrt{d})$. However,
$\phi^\prime$ has the disadvantage of not being defined over
$\F_q[x]$, since $g_{\phi^\prime}$ is not in $\F_q[x]$. We construct
an alternate $\phi$ that is isomorphic to $\phi^\prime$ but defined
over $\F_q[x]$. There is a notion of $J$-invariant for elliptic
modules (see e.g.~\cite{gek}); the $J$-invariant of $\phi^\prime$ is
\[ J_{\phi^\prime} := \frac{g_{\phi^\prime}^{q+1}}{\D_{\phi^\prime}} = 
d^{\frac{q+1}{2}}\left(1+d^{\frac{q-1}{2}}\right)^{q+1}. \]

Now let $\phi$ be the elliptic module defined by 
\[g_\phi = J_{\phi^\prime}, \quad  \D_\phi = J_{\phi^\prime}^{q} \]
The $J$-invariants of $\phi$ and $\phi'$ are the same, which implies that $\phi'$
and $\phi$ are isomorphic. Further, $\phi$ is defined over
$\F_q[x]$. In summary, $\phi$ has complex multiplication by
$\F_q(x)(\sqrt{d})$ and is defined over $\F_q[x]$.

\begin{remark}
When $q$ is even, the construction of Schweizer~\cite[Theorem
  6]{Schweizer97} gives the required elliptic modules with complex
multiplication. Pick an $(a,b)\in \F_q^\times \times \F_q$ at random
and set $d=ax+b$. The elliptic module $\phi$ with $J$-invariant
\[
J_{\phi} = \left(1+d+\ldots+d^{2^{m-2}}+d^{2^{m-1}}\right)^{q+1}/a
\]
has complex multiplication by the ring $\F_q[x][w]$ in the
Artin-Schreier extension $\F_q(x)(w)$, where $w^2+w=d$. We can take $\phi$ to
be the elliptic module defined by
\[g_\phi = J_{\phi}, \quad  \D_\phi = J_{\phi}^{q}, \]
ensuring $\phi$ is defined over $\F_q[x]$. For a further discussion
of characteristic 2, see Remark~\ref{rk:obstruction} below. 
\end{remark}

\noindent We now state our randomized algorithm to factor polynomials over
finite fields using elliptic modules with complex
multiplication. 

\begin{algorithm}[H]
  \caption{Polynomial factorization}
  \label{alg:factoring}
  \begin{algorithmic}[1]
    \REQUIRE Monic squarefree $f \in \F_q[x]$ of degree $n$
    \ENSURE The irreducible factors of $f$
    \STATE If $f$ is irreducible then output $f$ and return
    \STATE\label{step:fac-2}
    Remove the linear factors of $f$ and output them
    \STATE\label{step:fac-3}
    Pick $a \in \F_q$ uniformly at random and compute \\
    $d := x - a$,\\
    $J:=d^{\frac{q+1}{2}}\left(1+d^{\frac{q-1}{2}}\right)^{q+1} \bmod f$, \\
	$g_\phi:=J \bmod f$ and	$\Delta_\phi := J^q \bmod f$\\	
    \STATE\label{step:fac-split}
    Compute $\gamma := \gcd(\bar h_{\phi,f}, f)$ and recursively factor $\gamma$ and $f/\gamma$.
  \end{algorithmic}
\end{algorithm}


The irreducibility test in Step $1$ can be performed in
$O(n^{(\omega+1)/2}(\log n)^2 + \MM(n)\log q)$ operations in
$\F_q$~\cite{vzGG}, or $n^{1+\varepsilon} (\log q)^{2+o(1)}$ bit
operations. In Step $2$, all the linear factors of $f$ are found and
removed using a root finding algorithm; it takes an expected
$O(\MM(n)\log n \log(nq))$ operations in $\F_q$~\cite{vzGG}, or
$n^{1+\varepsilon} (\log q)^{2+o(1)}$ bit operations (this step needs
only be done once in the whole algorithm).

In Step $3$, we choose $a \in \F_q$ at random and construct a Drinfeld
module $\phi$ with complex multiplication by
$\F_q(x)(\sqrt{x-a})$. The primes that divide $\D_\phi$ are
precisely $\{(x-b), b \in \F_q, \sqrt{b-a} \notin \F_q\} \cup
\{x-a\}$. Hence, we might have run into issues of bad reduction if $f$
had linear factors; it is to prevent this that we performed root finding
in Step~\ref{step:fac-2}. The computation of $g_\phi$ and $\Delta_\phi$ takes
$O(\MM(n)\log q)$ operations in $\F_q$.

As mentioned in Eq.~\eqref{equ:fact-sep}, $\gcd(\bar h_{\phi, f}, f)$
is the product of all irreducible factors of degree greater than $1$
of $f$ that are supersingular with respect to $\phi$. Thus, our
algorithm separates the irreducible factors supported at the
supersingular primes from those supported at the ordinary primes. In
the following, we show that for an elliptic module chosen randomly as
in Step \ref{step:fac-3}, the splitting of $f$ in Step
\ref{step:fac-split} is random enough to ensure that the recursion
depth is $O(\log n)$.

We assume as in the statement of Theorem~\ref{theo:main-factor} that
$\bar h_{\phi, f}$ is computed using $H(n,q)$ operations in $\F_q$,
resp.\ $H^*(n,q)$ bit operations. Since the expected depth of the
recursion is logarithmic, using the superlinearity assumptions on $H$,
resp.\ $H^*$ made in that theorem, the runtime of the algorithm is then an expected
\begin{itemize}
\item $\tildO(H(n,q)  + n^{(\omega+1)/2} + \MM(n)\log q )$ operations in $\F_q$, or
\item $\tildO(H^*(n,q)) + n^{1+\varepsilon}(\log q)^{2+o(1)}$ bit operations;
\end{itemize}
this proves Theorem~\ref{theo:main-factor}. 

In the following lemma, we establish the claim above on the
probability of finding a nontrivial factor of $f$. For the lemma to
apply to the algorithm, we need to assume $\sqrt{q} \ge 12(n+2)$.
This assumption can be made without loss of generality: if $\sqrt{q} <
12(n+2)$, we might choose to factor over a slightly larger field
$\F_{q^\prime}$ where $q^\prime$ is the smallest power of $q$ such
that $\sqrt{q^\prime} > 12(n+2)$ and still recover the factorization
over $\F_q$ (c.f. \cite[Remark 3.2]{nar}). Further, the running times
are only affected by logarithmic factors in $n$. 

\begin{lemma}
  \label{splitting_lemma}
  Suppose that $f$ is not irreducible, without linear factors, and that $12(n+2) \le
  \sqrt{q}$. Let $\phi$ be an elliptic module with complex
  multiplication by the imaginary quadratic extension
  $\F_q(x)(\sqrt{x-a})$, where $a \in \F_q$ is chosen at random. Then,
  with probability at least $1/4$, the factor $h$ computed at
  Step~\ref{step:fac-split} of the algorithm is non-trivial.
\end{lemma}
\begin{proof}
  Suppose that $f$ admits two distinct monic irreducible factors $f_1$
  and $f_2$, of respective degrees $k_1,k_2 > 1$. We prove that with
  probability at least $1/4$, exactly one of $f_1$ or $f_2$ is
  supersingular with respect to a randomly chosen $\phi$, as
  constructed above.

  Since $k_1,k_2$ are greater than $1$, none of $\langle f_1\rangle,
  \langle f_2\rangle$ ramify in $\F_q(x)(\sqrt{x-a})$. Therefore, the
  probability that exactly one of $\langle f_1\rangle, \langle
  f_2\rangle$ is supersingular with respect to $\phi$ is the same as
  the probability that exactly one of them splits in
  $\F_q(x)(\sqrt{x-a})/\F_q(x)$.
	
  For $i = 1, 2$, let $K_i:=\F_q(x)(\alpha_i)$ be the hyperelliptic
  extension of $\F_q(x)$ obtained by adjoining a root $\alpha_i$ of
  $y^2-f_i$. Depending on the values of $q$ and $k_1,k_2$, by
  quadratic reciprocity over function fields \cite{carlitz1932},
  one of the following alternative holds:
  \begin{enumerate}[leftmargin = \leftmargini]
  \item[(a)] exactly one of $\langle f_1\rangle, \langle f_2\rangle$ splits
    in $\F_q(x)(\sqrt{x-a})$ if and only if $x-a$ splits in exactly
    one of $K_1, K_2$;
  \item[(b)] exactly one of $\langle f_1\rangle, \langle f_2\rangle$
    splits in $\F_q(x)(\sqrt{x-a})$ if and only if $x-a$     either splits  in
 $K_1$ and $K_2$, or splits in none of them.
  \end{enumerate}
  We are in case (a) if $q =1 \bmod 4$ or $k_1 = k_2 \bmod 2$, and
  in case (b) for all other values of the parameters. Since $f_1, f_2$ are distinct,
  $K_1$ and $K_2$ are linearly disjoint over $\F_q(x)$. Further,
  $K_1K_2$ is Galois over $\F_q(x)$ with
  \[ \gal(K_1K_2/\F_q(x)) \cong \gal(K_1/\F_q(x)) \times \gal(K_2/\F_q(x)) \cong \Z/2\Z \oplus \Z/2\Z. \]  
  For $\langle x-a\rangle $ to be neither totally split nor totally inert
  (case (a)), the Artin symbol 
  \[ (\langle x-a\rangle, K_1K_2/\F_q(x)) \in \gal(K_1K_2/\F_q(x)) \]
  has to be either $(0,1)$ or $(1,0)$ under the isomorphism
  $\gal(K_1K_2/\F_q(x)) \cong \Z/2\Z \oplus \Z/2\Z$; case (b) occurs
  when its value is either $(0,0)$ or $(1,1)$.  Applying the effective
  Chebotarev density theorem over function fields~\cite[Proposition~6.4.8]{FrJa08},
  for any $\eta=(\eta_1,\eta_2) \in \Z/2\Z \oplus \Z/2\Z$, 
  we obtain that the number $N_\eta$ of degree 1 primes in $\F_q(x)$ that are unramified in $K_1K_2$ 
  and of symbol $\eta$ satisfies
  \begin{align*}
  \left|N_\eta  - \frac{q}{4} \right| & \le  \frac 12 \left ( (4+g(K_1K_2)) q^{1/2} + 4 q^{1/4}  + (4 + g(K_1K_2)) \right)  
  \end{align*}
  where $g(K_1K_2)$ is the genus of $K_1K_2$. Taking into account the
  prime at infinity, we obtain that in case (a), the number $N_{\rm (a)}$ of degree one
  primes $\{\langle x-a\rangle, a \in \F_q\}$ that are neither totally
  inert nor totally split in $K_1K_2$, and that in case (b), the number $N_{\rm (b)}$ of degree
  one primes $\{\langle x-a\rangle, a \in \F_q\}$ that are either
  totally inert or totally split in $K_1K_2$ satisfy
  $$ N_{\rm (a)} \ge N_{(0,1)} + N_{(1,0)} -1, \quad N_{\rm (b)} \ge N_{(0,0)} + N_{(1,1)} -1.$$
  The previous inequalities imply
  \begin{align*}
N_{\rm (a)}, N_{\rm (b)} & \ge \frac 12 q - \left ( (4+g(K_1K_2)) q^{1/2} + 4 q^{1/4}  + (5 + g(K_1K_2)) \right)  \\
&\ge \frac 12 q - 3(4+g(K_1K_2)) q^{1/2}.
  \end{align*}
  Using the Riemann-Hurwitz genus
  formula~\cite[Theorem~7.16]{Rosen02}, we now prove $g(K_1K_2) \le
  k_1 + k_2 -2$. Indeed, after replacing $\F_q$ by a suitable
  algebraic extension $\F_{q'}$, the formula reads $2 g(K_1 K_2)-2 =
  4(0-2) + \sum_\alpha (e(\alpha)-1)$, where the sum is over all
  points $\alpha \in \P^3(\F_{q'})$ lying on the projective closure
  $C$ of the curve $V(y_1^2-f_1(x),y_2^2-f_2(x))$ at which the
  first-factor projection $C \to \P^1(\F_{q'})$ ramifies (here,
  $e(\alpha)$ is the ramification index). All roots of either $f_1$ or
  $f_2$ (which are all pairwise distinct) give two such points, each
  of them with ramification index $2$; the contribution of the point
  at infinity in $\P^1(\F_{q'})$ is at most $4-1=3$, so that we have
  $2 g(K_1 K_2)-2 \le -8 + 2(k_1 + k_2) + 3$, which implies our claim.

  Using the inequality $k_1 + k_2 \le n$, we deduce
  $$N_{\rm (a)}, N_{\rm (b)} \ge \frac 12 q - 3(n+2) q^{1/2},$$ 
  so that in either case, the probability of finding a non-trivial factorization
  of $f$ is at least $1/2- 3(n+2)/q^{1/2}$.
  Since $12(n+2)  \le q^{1/2}$, this probability is at least $1/4$.
\end{proof}

\begin{remark} \label{rk:obstruction}
There is an obstruction to an even characteristic analogue of Lemma
\ref{splitting_lemma} being true. Let $q$ be even and suppose $\phi$
is an elliptic module with complex multiplication by the ring
$\F_q[x][w]$ in an Artin-Schreier extension, where $(a,b) \in
\F_q^\times \times \F_q$ is chosen at random and $w^2+w=ax+b$. Let $f$
have distinct monic irreducible factors $f_1,f_2$ of respective
degrees $k_1>1,k_2>1$.  Since $k_1,k_2$ are greater than $1$, none of
$\langle f_1\rangle, \langle f_2\rangle$ ramify in
$\F_q(x)(w)$. Therefore, the probability that exactly one of $\langle
f_1\rangle, \langle f_2\rangle$ is supersingular with respect to
$\phi$ is the same as the probability that exactly one of them splits
in $\F_q(x)(w)/\F_q(x)$. The latter is equivalent to exactly one of
the two equations
\begin{equation}\label{eq:char_2_congruence}
  w^2+w = ax+b \bmod f_1(x),\ w^2+w = ax+b \bmod f_2(x)
\end{equation}
possessing a solution $w \in \F_q[x]$. Analogous to the proof of Lemma
\ref{splitting_lemma}, one may look to reciprocity laws for
Artin-Schreier extensions to phrase this condition in terms of
polynomial congruences involving $f_1$ and $f_2$ modulo
$ax+b$. However, the reciprocity law (due to Hasse \cite{Hasse34}, see
also \cite[Theorem 4.2]{Conrad}) dictates that the existence of
solutions to the equations \eqref{eq:char_2_congruence} depends not on
congruences modulo $ax+b$ but only on the $(k_1-1)^{th}$ degree
(resp. $(k_2-1)^{th}$) coefficient of $f_1$ (resp. $f_2$) and the
parity of the degrees $k_1$ and $k_2$. In particular, by
\cite[Theorem~4.2]{Conrad} or \cite[Example~4.1]{Conrad}, if $k_1$ and
$k_2$ have the same parity and the $(k_1-1)^{th}$ coefficient of $f_1$
equals the $(k_2-1)^{th}$ degree coefficient of $f_2$, then for no
$(a,b) \in \F_q^\times \times \F_q$ does exactly one of the equations
in \eqref{eq:char_2_congruence} has a solution. Hence, our algorithm
cannot distinguish the factors $f_1$ and $f_2$. It remains an open
question as to if our algorithmic framework can be modified to handle
this situation. 

On the other hand, we may argue that this is the only obstruction. Say
$f_1 = x^{k_1}+c_1x^{k_1-1}+\cdots$ and $f_2 =
x^{k_2}+c_2x^{k_2-1}+\cdots$, where $c_1 \neq c_2$. Pick an elliptic
module $\phi$ with complex multiplication by the ring $\F_q[x][w]$ in
an Artin-Schreier extension, where $a \in \F_q^\times$ is chosen at
random and $w^2+w=ax$. By \cite[Example 4.1]{Conrad}, $w^2+w = x \bmod
f_1(x)$ has a solution $w \in \F_q[x]$ if and only if
${\rm Tr}_{\F_q/\F_2}(ac_1) = 0$, where ${\rm Tr}_{\F_q/\F_2}$ is the trace from
$\F_q$ to $\F_2$. Likewise, $w^2+w = x \bmod f_1(x)$ has a solution $w
\in \F_q[x]$ if and only if ${\rm Tr}_{\F_q/\F_2}(ac_2) = 0$. For our
algorithm to separate $f_1$ and $f_2$, it suffices for
${\rm Tr}_{\F_q/\F_2}(a(c_1+c_2))$ to equal $1$ with constant
probability. Since $c_1 \neq c_2$, $c_1+c_2 \neq 0$ and
${\rm Tr}_{\F_q/\F_2}(a(c_1+c_2))$ does equal $1$ with probability at least
$1/2$.
\end{remark}


\section{Implementation and example}
\label{sec:impl}

We implemented Algorithms~\ref{alg:hasse-inv} and~\ref{alg:factoring}
in C++ using NTL~\cite{shoup2001ntl}; the implementation source code
can be found at
\url{https://github.com/javad-doliskani/supersingular_drinfeld_factoring}. After
running the algorithm on $10^3$ random input polynomials, we observed
that the splitting $\gcd(\bar h_{\phi, f}, f)$ is almost always
nontrivial when $f$ is reducible; this confirms the behavior expected
from theory.

We also observed that in practice one does not need to impose the
condition $\sqrt{q} \ge 12(n+2)$ for splitting a polynomial of degree
$n$ over $\F_q$; this remedies the need for working in an auxiliary
extension over $\F_q$. Recall that Lemma~\ref{splitting_lemma} uses
that assumption to bound from below the probability of finding a
non-trivial factorization; we expect that the recursion depth be
logarithmic in $n$ also for smaller values of $q$.  At worst, one may
first attempt splitting $f$ over $\F_q$ and in case of repeated failures,
we may then switch to factorization over an extension; this
was not implemented.

In terms of runtime, our algorithm is slower that NTL's built-in
factorization routine, by what is roughly a constant factor; this is
not entirely surprising, since working with $2\times 2$ matrices
induces a non-negligible overhead for our algorithm. To show the
behaviour of the algorithm in practice, we give in this section an
example output of Algorithm \ref{alg:factoring} on input a small
degree polynomial.

\begin{example*}
Consider the randomly selected squarefree polynomial 
\[ f = 2 + 6x + 5x^3 + 4x^4 + 6x^5 + 2x^7 + 3x^8 + 3x^9 + x^{10} \]
in $\F_7[x]$. The algorithm starts by checking for linear factors; $f$ has none. So, in the 
next step it generates the random supersingular elliptic module $\phi$ given by
choosing $d=1+x$ and computing
\[
\begin{array}{rll}
	g_\phi & = & 3 + 3x + 5x^3 + x^4 + x^5 + x^6 + 6x^7 + 5x^8 + 5x^9,\\
        \Delta_\phi&=&4x+4x^2+5x^3+3x^4+5x^5+6x^6+5x^7+4x^8+2x^9.
\end{array}
\]
This means $\phi$ has complex multiplication by the imaginary
quadratic extension $\F_7(x)(\sqrt{x + 1})$ of $\F_7(x)$. After
computing $\bar h_{\phi, f}$, we
find $\gamma = \gcd(\bar h_{\phi, f}, f) = 1 + 4x + x^2 + 4x^3 + x^4$; this is
the product of the factors of $f$ that are supersingular with respect
to $\phi$. Now, $\gamma$ does not pass the irreducibility check, so the
process is repeated for $\gamma$. The next randomly generated supersingular
elliptic module is built by choosing $d=5+x$ and computing
\[
\begin{array}{rll}
	g_{\tilde\phi} & = & 6 + x + 5x^2 + 5x^3,\\
	\Delta_{\tilde \phi} & = & 4+x + 6x^2 + 4x^3.
\end{array}
\]
We then split $\gamma$ as $ \tilde \gamma = \gcd(\bar h_{\tilde \phi,
  \gamma}, \gamma) = 4 + 6x + x^2$, and $\gamma/\tilde \gamma = 2 + 5x
+ x^2$ which are both irreducible; again, this means that $\tilde \gamma$
is supersingular and $\gamma/\tilde \gamma$ is ordinary with respect to $\tilde
\phi$. Finally, $f / \gamma = 2 + 5x + 6x^2 + 3x^3 + 6x^4 + 6x^5 + x^6$
which is also irreducible. The complete factorization of $f$ is
then \[ f(x) = (4 + 6x + x^2)(2 + 5x + x^2)(2 + 5x + 6x^2 + 3x^3 +
6x^4 + 6x^5 + x^6).\]
\end{example*}

\paragraph*{Acknowledgements.} We thank the reviewer of a first version
of this paper for useful remarks, which led us to this revised version. Doliskani is partially 
supported by NSERC, CryptoWorks21, and Public Works and Government Services Canada. Narayanan is 
supported by NSF grant \#CCF-1423544, Chris Umans' Simons Foundation Investigator grant and 
European Union's H2020 Programme under grant agreement number ERC-669891.

\bibliographystyle{siamplain}
\bibliography{references}
\end{document}